\documentclass[a4paper,UKenglish,cleveref, autoref, thm-restate]{lipics-v2021}

\usepackage{nicefrac}
\usepackage{algorithm2e}
\RestyleAlgo{ruled}
\nolinenumbers

\title{Laminar Matroid Secretary: Greedy Strikes Back} 

\author{Zhiyi Huang}{Department of Computer Science, The University of Hong Kong, Hong Kong \and \url{https://i.cs.hku.hk/~zhiyi/} }{hzhiyi.tcs@gmail.com}{https://orcid.org/0000-0003-2963-9556}{}

\author{Zahra Parsaeian}{University of Freiburg, Germany}{zahrap@cs.uni-freiburg.de}{}{}

\author{Zixuan Zhu}{Department of Computer Science, The University of Hong Kong, Hong Kong}{zhuzixua@connect.hku.hk}{https://orcid.org/0009-0009-1099-468X}{}

\authorrunning{Z.~Huang, Z.~Parsaeian, and Z.~Zhu} 

\Copyright{Jane Open Access and Joan R. Public} 



\category{} 

\relatedversion{} 

\funding{This work is supported by General Research Fund from Research Grants Council of Hong Kong (Grand No. 17211723).}

\EventEditors{Timothy Chan, Johannes Fischer, John Iacono, and Grzegorz Herman}
\EventNoEds{4}
\EventLongTitle{32nd Annual European Symposium on Algorithms (ESA 2024)}
\EventShortTitle{ESA 2024}
\EventAcronym{ESA}
\EventYear{2024}
\EventDate{September 2--4, 2024}
\EventLocation{Royal Holloway, London, United Kingdom}
\EventLogo{}
\SeriesVolume{308}
\ArticleNo{49}

\begin{document}

\maketitle

\begin{abstract}
    We show that a simple greedy algorithm is $4.75$-competitive for the Laminar Matroid Secretary Problem, improving the $3\sqrt{3} \approx 5.196$-competitive algorithm based on the forbidden sets technique (Soto, Turkieltaub, and Verdugo, 2018). 
\end{abstract}

\section{Introduction}
Decision-making under uncertainty has always been a central research topic in Computer Science and Operations Research. 
A classical problem, and one of the oldest in this area, is the Secretary Problem~\cite{dynkin1963optimal, Ferguson1989WhoST, Lindley1961DynamicPA}.
Here a decision-maker wants to hire the best secretary out of $n$ applicants.
The applicants are interviewed one at a time by a random order.
After each applicant's interview, the decision-maker learns how he/she ranks compared to the previous applicants, and must then immediately decide whether to hire him/her. 
The goal is to maximize the probability of hiring the best secretary.
This problem has been extensively studied in the literature and several optimal $e$-competitive algorithms have been proposed.
A well-known elegant solution by Dynkin~\cite{dynkin1963optimal} goes as follows:
reject the first $\frac{n}{e}$ applicants, and then hire the first subsequent applicant that is better than all previous applicants.

Babaioff, Immorlica, and Kleinberg~\cite{DBLP:conf/soda/BabaioffIK07} formulated a generalization called the Matroid Secretary Problem, which considers hiring $r \ge 1$ secretaries subject to a matroid constraint. 
From now on, we will refer to the secretaries as elements. We are given a matroid $\mathcal{M}=(E, \mathcal{I})$, where $E$ is a finite set of elements and $\mathcal{I} \subseteq 2^E$ is a collection of \emph{independent sets} of elements from $E$.
We will defer the definition of general matroids to the next section, and refer interested readers to the textbook by Oxley~\cite{oxley2006matroid} for an in-depth exposition of matroid theory. 
For now, readers may think of $\mathcal{I}$ as the collection of all subsets with size at most $r$, known as the $r$-uniform matroid, as a running example.
Each element $e \in E$ has a positive weight $w_e$.
Elements' weights are revealed to the algorithm one at a time by a random order.
After observing an element's weight, the algorithm must immediately decide whether to select the element.
The goal is to select an independent set that maximizes the sum of the weights of the elements in it.
An algorithm is $c$-competitive if the expected total weight of its selected elements is at least a $\frac{1}{c}$ fraction of the total weight of the optimal solution in hindsight.
Babaioff et al.~\cite{DBLP:conf/soda/BabaioffIK07} conjectured the existence of a constant competitive algorithm for the Matroid Secretary Problem, which has been a major open question in the literature for more than a decade.
The best competitive ratio so far is $\mathcal{O}(\log\log r)$, due to Lachish~\cite{DBLP:conf/focs/Lachish14} and Feldman, Svensson, and Zenklusen~\cite{DBLP:journals/mor/FeldmanSZ18}.

Although the general Matroid Secretary Problem remains elusive, much progress has been made toward designing constant competitive algorithms for special matroids.
This paper focuses on the laminar matroids.
A \emph{laminar family} $\mathcal{F}$ over elements $E$ is a collection of subsets of elements such that any pair of distinct subsets $A$ and $B$ are either related by containment or disjoint, i.e., either $A \subset B$, or $B \subset A$, or $A \cap B = \emptyset$.
A \emph{laminar matroid} puts a capacity $c(B)$ on each subset $B \in \mathcal{F}$;
a subset of elements is independent if it has at most $c(B)$ elements in each $B \in \mathcal{F}$. 
We call it the Laminar Matroid Secretary Problem.

Im and Wang~\cite{DBLP:conf/soda/ImW11} gave the first constant competitive algorithm for laminar matroids, with a competitive ratio of $\frac{16000}{3}$.
Jaillet, Soto, and Zenklusen~\cite{DBLP:conf/ipco/JailletSZ13} proposed an improved $3\sqrt{3}e \approx 14.12$-competitive algorithm by reducing it to the problem with partition matroids and running an $e$-competitive algorithm on the resulting problem.
Ma, Tang, and Wang~\cite{DBLP:journals/mst/Ma0W16} showed that directly running a greedy algorithm is $9.6$-competitive.
Finally, Soto, Turkieltaub, and Verdugo~\cite{DBLP:journals/mor/SotoTV21} introduced the forbidden sets technique and applied it to various matroids.
For laminar matroids, they combined it with the reduction to partition matroids by Jaillet et al.~\cite{DBLP:conf/ipco/JailletSZ13} to yield the best existing $3\sqrt{3} \approx 5.196$-competitive ratio.

This paper considers a simple greedy algorithm that is different from the previous greedy algorithm by Ma et al.~\cite{DBLP:journals/mst/Ma0W16}.
We show that our algorithm is $4.75$-competitive for the Laminar Matroid Secretary Problem, improving upon the best existing result by Soto et al.~\cite{DBLP:journals/mor/SotoTV21}.

\section{Preliminaries}

We will assume that the elements' weights are distinct, which simplifies the arguments by avoiding tie-breaking cases.
This is without loss of generality as we can randomly perturb the weights with negligible Gaussian noise.

For notational simplicity, we write $S + e$ for $S \cup \{ e \}$ and $S - e$ for $S \setminus \{ e \}$.

\begin{definition}
A matroid $\mathcal{M}=(E, \mathcal{I})$ consists of a set of elements $E$ and collection of subsets of elements $\mathcal{I}$ called the independent sets, with the following properties:
\begin{itemize}
\item \textbf{Non-empty Property:~}
$\varnothing \in \mathcal{I}$;
\item \textbf{Hereditary Property:~}
For any $S \subset T \subseteq E$, if $T \in \mathcal{I}$ then $S \in \mathcal{I}$;
\item \textbf{Augmentation Property:~}
For any $S, T \in \mathcal{I}$ such that $|S| > |T|$, there exists an element $e \in S \setminus T$ such that $T + e \in \mathcal{I}$.
\end{itemize}
\end{definition}

Another useful property of matroids, even considered to be the fourth axiom by some, is the following \textbf{Exchange Property}.

\begin{lemma}[c.f., Schrijver~\cite{Schrijver}, Corollary 39.12a]
    For any $S, T \in \mathcal{I}$ such that $|S| = |T|$, for any $e \in S \setminus T$, there is an element $e' \in T \setminus S$ such that $S - e + e', T - e' + e \in \mathcal{I}$.
\end{lemma}

Let $OPT$ denote the independent set that maximizes the sum of its elements' weights, which is unique when all the elements have distinct weights.
We will use a well-known property of matroids, whose proof we include for self-containedness.

\begin{lemma}
    \label{lem:basic-matroid}
    For any matroid $\mathcal{M}=(E, \mathcal{I})$, and any subset of elements $S \subseteq E$, the elements in $OPT \cap S$ also belong to the optimal solution of the sub-matroid restricted to the elements in $S$, i.e., $\mathcal{M}_S=(S, \mathcal{I} \cap 2^S)$.
\end{lemma}

\begin{proof}
    Let $OPT(S)$ denote the optimal solution of sub-matroid $\mathcal{M}_S$.
    Suppose to the contrary that there exists an element $e \in OPT \cap S$ such that $e \not \in OPT(S)$.

    We first show that there is an element $e' \in OPT(S)$ such that:
    $$
        OPT - e + e' ~,~ OPT(S) - e' + e ~\in~ \mathcal{I}
        ~.
    $$
    
    If $|OPT(S)| = |OPT|$, it follows by the exchange property.
    
    Next suppose that $|OPT(S)| < |OPT|$.
    By repeatedly applying the augmentation property, we can find elements $e_1, e_2, \dots, e_k \in OPT \setminus OPT(S)$, where $k = |OPT| - |OPT(S)|$, such that $OPT(S) + e_1 + \dots + e_k \in \mathcal{I}$.
    Further, the optimality of $OPT(S)$ and the fact that $e \not \in OPT(S)$ indicate $e_1, \dots, e_k \ne e$.
    We next apply the exchange property to $OPT(S) + e_1 + \dots + e_k$ and $OPT$ to conclude that there is an element $e' \in (OPT(S) + e_1 + \dots + e_k) \setminus OPT = OPT(S) \setminus OPT$ such that $OPT - e + e' \in \mathcal{I}$ and $OPT(S) + e_1 + \dots + e_k - e' + e \in \mathcal{I}$.
    Finally, the latter implies $OPT(S) - e' + e \in \mathcal{I}$ by the hereditary property.

    Given such an element $e'$, we can now derive a contradiction from the optimality of $OPT$ and $OPT(S)$:
    the optimality of $OPT$ implies $w_e > w_{e'}$, while the optimality of $OPT(S)$ indicates $w_{e'} > w_e$ (recall that $e \not \in S$).
\end{proof}

\section{Greedy Algorithm}
\label{sec:algorithm}

This section introduces our greedy algorithm.
Instead of letting the elements arrive one by one in $n$ discrete time steps, we consider an equivalent continuous time model, in which each element's arrival time is drawn independently and uniformly between $0$ and $1$.%
\footnote{To simulate the continuous time model in the discrete-time model, we can first draw $n$ arrival times $t_1 < t_2 < \dots < t_n$ between $0$ and $1$, and let $t_i$ be the arrival time of the $i$-th element.}
The continuous time model implies independence of different elements' arrival times, which leads to a simpler analysis.

We further define the following notations:
\begin{itemize}
    \item Let $E(t)$ denote the set of elements that arrive from time $0$ to $t$ (inclusive).
    \item Let $OPT(t)$ be the optimal solution with respect to elements in $E(t)$.
    \item Let $OPT = OPT(1)$ denote the optimal solution with respect to all elements.
\end{itemize}

For some time threshold $t_0$ to be determined, our greedy algorithm rejects all elements that arrive before $t_0$, like the classical algorithm by Dynkin~\cite{dynkin1963optimal}.
For an element that arrives at time $t_0 < t \le 1$, the algorithm selects it under two conditions:
it belongs to the offline optimal solution with respect to the arrived elements in $E(t)$;
and adding it to the selected elements gives an independent set.

\begin{algorithm}[hbt!]
\SetKwInOut{Input}{input}\SetKwInOut{Output}{output}
\caption{Greedy}\label{alg:greedy}
\Input{~matroid $\mathcal{M}=(E, \mathcal{I})$ and a threshold time $t_0$}
\Output{~set of selected elements $ALG$, initialially set to be the empty set $\varnothing$.}
\For{each element $e$ that arrives at time $0 \leq t \leq 1$}
{\textbf{let} $ALG \gets ALG \cup e$ \textbf{if}
    (1) $t > t_0$,
    (2) $e \in OPT(t)$, and
    (3) $ALG \cup e \in \mathcal{I}$.
}
\end{algorithm}

\cref{alg:greedy} is computationally efficient.
In particular, $OPT(t)$ at any time $t$ can be computed in polynomial time by processing the elements in $E(t)$ by descending order of weights, and selecting each element whenever independence is still satisfied.

We next compare \cref{alg:greedy} with the $9.6$-competitive greedy algorithm by Ma et al.~\cite{DBLP:journals/mst/Ma0W16}.
The only difference is the second condition for selecting an element $e$.
Their algorithm checks if $e$ is in the optimal solution with respect to $E(t_0) \cup \{e\}$, i.e., the elements that arrive before the threshold time $t_0$ and element $e$.
By contrast, our algorithm examines $OPT(t)$, i.e., the optimal solution with respect to all arrived elements.
By \cref{lem:basic-matroid}, elements in the optimal solution $OPT$ would pass both criteria regardless of their arrival time.
Ours is stricter, and hence, is less likely to select elements outside $OPT$.

This natural idea of checking local optimality was applied in Babaioff et al's \cite{DBLP:journals/jacm/BabaioffIKK18} $e$-competitive \textit{virtual algorithm} for uniform matroids, which maintains a reference set at each time $t$ to record the current top-$k$ elements and selects an element if it belongs to the reference set on arrival.
Same idea also appeared in the $e$-competitive online algorithm for weighted bipartite matching given by Kesselheim et al \cite{DBLP:conf/esa/KesselheimRTV13}.

\section{Competitive Analysis}
\label{sec:analysis}

\begin{theorem}\label{thm:main}
    \cref{alg:greedy} (with $t_0 = 0.7$) is $4.75$ probability-competitive for the Laminar Matroid Secretary Problem.
\end{theorem}

\begin{proof}
Consider any element $e^{*} \in OPT$.
It suffices to prove that \cref{alg:greedy} selects $e^{*}$ with probability at least $\frac{1}{4.75}$.

Since $e^{*}$ would satisfy the second condition regardless of its arrival time, i.e., $e^{*} \in OPT(t)$ (\cref{lem:basic-matroid}), we consider the case when it arrives at time $t_0 < t^* \le 1$, and lower bound the probability that all subsets containing $e^{*}$ in the laminar family, denoted as $B_1 \subset B_2 \subset \dots \subset B_k$, still have unused capacities.
We can assume without loss of generality that $c(B_i) < c(B_{i+1})$.
Otherwise, we can remove $B_i$ and its capacity and still have the same laminar matroid.

We first bound the failure probability for each subset $B_i$ separately.
That is, we consider the probability that the algorithm selects $c(B_i)$ elements from time $t_0$ to $t^*$.
To avoid the complex dependence across different subsets, we will analyze a relaxed event.
Consider the sub-laminar matroid restricted to elements in $B_i$, denoted as $\mathcal{M}_{B_i}$.
We say that an element $e \in B_i$ arriving at $0 \le t \le 1$ \textit{qualifies} (with respect to $B_i$) if it belongs to the optimal solution with respect to sub-matroid $\mathcal{M}_{B_i \cap E(t)}$, denoted as $OPT_i(t)$.
By \cref{lem:basic-matroid} with matroid $\mathcal{M}_{E(t)}$ and subset $B_i \cap E(t)$, any element $e \in B_i$ selected by the algorithm, in particular, any $e \in OPT(t) \cap B_i$, belongs to $OPT_i(t)$ and thus qualifies.
Hence, it suffices to upper bound the probability that at least $c(B_i)$ elements in $B_i$ qualify from time $t_0$ to $t$.

To do so, we next establish two lemmas that characterize the arrival times of qualified elements.
For ease of presentation, we will assume without loss of generality that $|OPT_i(t)| = c(B_i)$ at all time, by adding infinitely many dummy elements with negligible weights.

\begin{lemma}\label{lem:last-qualified}
Fix any time $0 < t \le 1$ and subset $B_i$ in the laminar family. 
Let $t_{-1}$ be the arrival time of the last qualified element of $B_i$ from time $0$ to $t$ (exclusive).
Then, $\ln \frac{t}{t_{-1}}$ follows the exponential distribution with rate parameter $c(B_i)$, i.e.,
$\Pr \big[ \ln \frac{t}{t_{-1}} \le x \big] = 1 - e^{-c(B_i)x}$.
\end{lemma}

\begin{proof}[Proof of \cref{lem:last-qualified}]
    We first rewrite the cumulative distribution function:
    $$
        \Pr \big[ \ln \tfrac{t}{t_{-1}} \le x \big] = 1- \Pr \big[ \ln \tfrac{t}{t_{-1}} > x \big] = 1- \Pr \big[ t_{-1} < e^{-x} t \big]
        ~.
    $$

    It remains to show that:
    $$
         \Pr \big[ t_{-1} < e^{-x} t \big] = e^{-c(B_i)x}
        ~.
    $$

    Let $E(t^-)$ be the elements that arrive from time $0$ to $t$ (exclusive).
    Accordingly, let $OPT_i(t^-)$ be the optimal solution with respect to sub-matroid $\mathcal{M}_{E(t^-) \cap B_i}$.
    We will prove that $t_{-1}$ is the arrival time of the last element in $OPT_i(t^-)$.
    In other words, the last element in $OPT_i(t^-)$ is the last qualified element.
    
    On the one hand, by \cref{lem:basic-matroid} with matroid $\mathcal{M}_{E(t^-) \cap B_i}$, any element in $OPT_i(t^-)$ must qualify.
    On the other hand, consider any element $e'$ that arrives at some time $t'$ after all elements in $OPT_i(t^-)$ arrive.
    This means $OPT_i(t^-) \subseteq E(t')$.
    Further by \cref{lem:basic-matroid} with matroid $\mathcal{M}_{E(t^-) \cap B_i}$, we get that $OPT_i(t^-) = OPT_i(t^-) \cap E(t') \subseteq OPT_i(t')$. 
    This implies $OPT_i(t^-) = OPT_i(t')$ and thus $e' \notin OPT_i(t')$.
    That is, $e'$ cannot qualify.

    The lemma now follows because the probability that all $c(B_i)$ elements in $OPT_i(t^-)$ arrive before $e^{-x} t$ equals $e^{-c(B_i)x}$, where these elements' arrival times distribute independently and uniformly between time $0$ and $t$ (exclusive).
\end{proof}

\begin{lemma}\label{lem:gamma}
    Fix any time $t$ and subset $B_i$ in the laminar family. Let $t_{-k}$ be the arrival time of the $k$-th last qualified element of $B_i$ from time $0$ and $t$ (exclusive).
    Then, $\ln \frac{t}{t_{-c(B_i)}}$ follows the Gamma distribution with shape and rate parameters equal to $c(B_i)$.
\end{lemma}

To be self-contained, we include below the cumulative function of Gamma distribution with shape parameter $\alpha$ and rate parameter $\beta$:
    $$
        F(x; \alpha, \beta) = \int_0^x \frac{y^{\alpha-1} e^{-\beta y} \beta^\alpha}{(\alpha-1)!} \,\textrm{d}y
        ~.
    $$
We refer interested readers to Chapter $8.4$ of Blitzstein and Hwang~\cite{blitzstein2019introduction} for further information about Gamma distributions.

\begin{proof}[Proof of \cref{lem:gamma}]
Observe that 
    $$
        \ln \frac{t}{t_{-c(B)}} = \ln \frac{t}{t_{-1}} + \ln \frac{t_{-1}}{t_{-2}} + \cdots + \ln \frac{t_{-c(B_i)+1}}{t_{-c(B_i)}}
        ~.
    $$

By \cref{lem:last-qualified}, $\ln \frac{t}{t_{-1}}$ follows the exponential distribution with rate parameter $c(B_i)$.
Further note that conditioned on $t_{-1}$ and the subset of elements that arrive before $t_{-1}$, these elements independently and uniformly arrive from $0$ to $t_{-1}$ (exclusive).
Hence, by \cref{lem:last-qualified} with $t = t_{-1}$, we get that $\ln \frac{t_{-1}}{t_{-2}}$ follows the exponential distribution with rate parameter $c(B_i)$, independent to $\ln \frac{t}{t_{-1}}$.
Repeating this argument shows that $\ln \frac{t}{t_{-c(B_i)}}$ is the sum of $c(B_i)$ independent exponential random variables with rate parameter $c(B_i)$, and thus follows the Gamma distribution with shape and rate parameters equal to $c(B_i)$.%
\footnote{The fact that the sum of exponential distributions with the same parameter forms a Gamma distribution can be found, e.g., in Theorem 8.4.3 of Blitzstein and Hwang~\cite{blitzstein2019introduction}. }
\end{proof}

\bigskip

We are now ready to upper bound the probability that at least $c(B_i)$ elements in $B_i$ qualify from time $t_0$ to $t^*$, because $B_i$ has at least $c(B_i)$ qualified elements arrive from time $t_0$ to $t^*$ if and only if $t_{-c(B_i)} \ge t_0$. 
As a corollary of \cref{lem:gamma}, the probability of both events equals $F(\ln \frac{t}{t_0}; c(B_i), c(B_i))$. 
Further by union bound, and by that $c(B_1) < c(B_2) < \dots < c(B_k)$, with probability at least $1 - \sum_{i=1}^\infty F \big( \ln\frac{t^*}{t_0}, i, i \big)$ we have that less than $c(B_i)$ qualified elements arrive from time $t_0$ to $t$ for any $B_i$.
Hence, \cref{alg:greedy} selects element $e^{*}$ with probability at least:
    $$
        \int_{t_0}^1  ~ 1 - \sum_{i=1}^\infty F \Big( \ln\frac{t^*}{t_0}, i, i \Big)  \,\textrm{d} t^*
        ~.
    $$

The theorem follows as one can verify numerically that it is at least $\frac{1}{4.75}$ when $t_0 = 0.7$.

For completeness, we next present the details of the numerical verification.
We divide the infinite sum of cumulative distribution functions into two parts:
    $$
        \sum_{i=1}^{3000} F \Big( \ln\frac{t^*}{t_0}, i, i \Big) 
    \quad\mbox{and}\quad
    \sum_{i=3001}^\infty F \Big( \ln\frac{t^*}{t_0}, i, i \Big)
        ~.
    $$

The first part is finite and hence can be calculated numerically.%
\footnote{The numerical calculations are conducted in Python. Interested readers may find the source codes through: https://github.com/ZixuannnZhu/Laminar-Matroid-Secretary. \label{git_link}}
We next upper bound the second part.
Recall that the cumulative distribution function of Gamma distribution is:
    $$
        F(x,i,i)=\int_{0}^{x} \frac{y^{i-1} e^{-iy} i^i}{(i-1)!} \textrm{d} y = \int_{0}^{x^i} \frac{i^i}{i!} e^{-iy}\textrm{d} y^i
        ~.
    $$

By $e^{-iy} \le 1$, this integral is bounded from above by $\frac{x^i i^i}{i!}$, which is further upper bound by $(xe)^i$ because $i! \ge i^i e^{-i}$ by Stirling's approximation.
Hence:
    $$
        \sum_{i=3001}^\infty F \Big( \ln\frac{t^*}{t_0}, i, i \Big) < \sum_{i=3001}^\infty \Big(e \ln\frac{t^*}{t_0} \Big)^{i} 
        ~,
    $$
which is a geometric sequence. 
The convergence condition requires that $e \ln \frac{t^*}{t_0} < 1$, where $t^* \in [0,1]$. 
This restricts our choice of $t_0$ to be greater than $e^{-1/e} \approx 0.692$. 
On the other hand, delaying the threshold time to a very late stage will lead to significant loss by missing all the optimal elements arriving before $t_0$. Numerical calculations suggest choosing $t_0=0.7$ to utilize the success probability,\textsuperscript{\ref{git_link}} with which we have:
    $$
        \sum_{i=3001}^\infty \Big(e \ln\frac{t^*}{t_0} \Big)^{i} < \sum_{i=3001}^\infty \Big(e \ln\frac{1}{0.7} \Big)^{i} < 10^{-38}
        ~.
    $$


Hence, the probability that \cref{alg:greedy} selects element $e^{*}$ is at least:
    $$
        \int_{0.7}^1  1 - \sum_{i=1}^\infty F \Big( \ln\frac{t^*}{0.7}, i, i \Big) \textrm{d} t^* > \int_{0.7}^1  1 - \sum_{i=1}^{3000} F \Big( \ln\frac{t^*}{0.7}, i, i \Big) - 10^{-38} \textrm{d} t^* > \frac{1}{4.75}
        ~.
    $$
\end{proof}

\section{Conclusion}

We conclude the article with a few remarks.
First, our algorithm is $4.75$ probability competitive, i.e., any element from the optimal solution is selected with probability at least $\frac{1}{4.75}$.
This is slightly stronger than the standard notion of competitive analysis, a.k.a., utility competitiveness in the context of Matroid Secretary Problem.

Further, our algorithm works in the ordinal model posed by Soto et al.~\cite{DBLP:journals/mor/SotoTV21}, where the decision-maker does not have numerical information about elements' weights, but only knows their relative ranks.

Finally, some readers may notice that our analysis has underestimated the algorithm's probability of selecting an element from the optimal solution, because the union bound sums over all possible capacities from $1$ to infinity, instead of, say, to the rank of the matroid.
Nevertheless, the ratio does not change significantly even if we only sum to $10$. 
The competitive ratio given by our method for different ranks, with $t_0$ optimized for each rank, is shown in the following figure.\textsuperscript{\ref{git_link}}
When rank $=1$, \cref{alg:greedy} is $e$ probability competitive by taking $t_0 = \frac{1}{e}$, matching the results of Dynkin~\cite{dynkin1963optimal} for the original Secretary Problem.

\clearpage

\begin{figure}[h]
\centering
\includegraphics{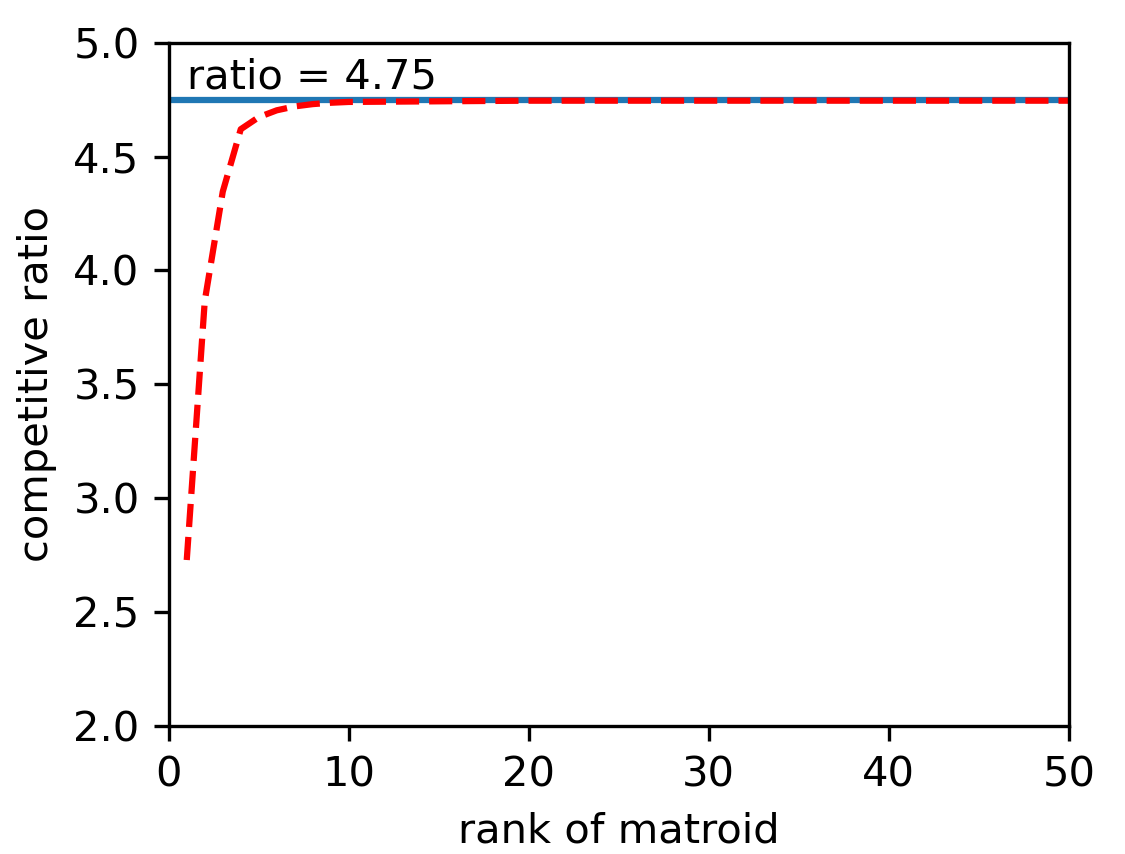}
\end{figure}



\bibliography{matroid-secretary}

\begin{thebibliography}{10}

\bibitem{DBLP:journals/jacm/BabaioffIKK18}
Moshe Babaioff, Nicole Immorlica, David Kempe, and Robert Kleinberg.
\newblock Matroid secretary problems.
\newblock {\em J. {ACM}}, 65(6):35:1--35:26, 2018.
\newblock \href {https://doi.org/10.1145/3212512} {\path{doi:10.1145/3212512}}.

\bibitem{DBLP:conf/soda/BabaioffIK07}
Moshe Babaioff, Nicole Immorlica, and Robert Kleinberg.
\newblock Matroids, secretary problems, and online mechanisms.
\newblock In {\em Proceedings of the Eighteenth Annual {ACM-SIAM} Symposium on Discrete Algorithms, {SODA} 2007, New Orleans, Louisiana, USA, January 7-9, 2007}, pages 434--443, 2007.
\newblock URL: \url{http://dl.acm.org/citation.cfm?id=1283383.1283429}.

\bibitem{blitzstein2019introduction}
Joseph~K Blitzstein and Jessica Hwang.
\newblock {\em Introduction to probability}.
\newblock Crc Press, 2019.

\bibitem{dynkin1963optimal}
Evgenii~B Dynkin.
\newblock Optimal choice of the stopping moment of a markov process.
\newblock In {\em Doklady Akademii Nauk}, volume 150, pages 238--240. Russian Academy of Sciences, 1963.

\bibitem{DBLP:journals/mor/FeldmanSZ18}
Moran Feldman, Ola Svensson, and Rico Zenklusen.
\newblock A simple \emph{O}(log log(rank))-competitive algorithm for the matroid secretary problem.
\newblock {\em Math. Oper. Res.}, 43(2):638--650, 2018.
\newblock URL: \url{https://doi.org/10.1287/moor.2017.0876}, \href {https://doi.org/10.1287/MOOR.2017.0876} {\path{doi:10.1287/MOOR.2017.0876}}.

\bibitem{Ferguson1989WhoST}
Thomas~S. Ferguson.
\newblock Who solved the secretary problem.
\newblock {\em Statistical Science}, 4:282--289, 1989.

\bibitem{DBLP:conf/soda/ImW11}
Sungjin Im and Yajun Wang.
\newblock Secretary problems: Laminar matroid and interval scheduling.
\newblock In {\em Proceedings of the Twenty-Second Annual {ACM-SIAM} Symposium on Discrete Algorithms, {SODA} 2011, San Francisco, California, USA, January 23-25, 2011}, pages 1265--1274, 2011.
\newblock \href {https://doi.org/10.1137/1.9781611973082.96} {\path{doi:10.1137/1.9781611973082.96}}.

\bibitem{DBLP:conf/ipco/JailletSZ13}
Patrick Jaillet, Jos{\'{e}}~A. Soto, and Rico Zenklusen.
\newblock Advances on matroid secretary problems: Free order model and laminar case.
\newblock In {\em Integer Programming and Combinatorial Optimization - 16th International Conference, {IPCO} 2013, Valpara{\'{\i}}so, Chile, March 18-20, 2013. Proceedings}, pages 254--265, 2013.
\newblock \href {https://doi.org/10.1007/978-3-642-36694-9\_22} {\path{doi:10.1007/978-3-642-36694-9\_22}}.

\bibitem{DBLP:conf/esa/KesselheimRTV13}
Thomas Kesselheim, Klaus Radke, Andreas T{\"{o}}nnis, and Berthold V{\"{o}}cking.
\newblock An optimal online algorithm for weighted bipartite matching and extensions to combinatorial auctions.
\newblock In {\em Algorithms - {ESA} 2013 - 21st Annual European Symposium, Sophia Antipolis, France, September 2-4, 2013. Proceedings}, pages 589--600, 2013.
\newblock \href {https://doi.org/10.1007/978-3-642-40450-4\_50} {\path{doi:10.1007/978-3-642-40450-4\_50}}.

\bibitem{DBLP:conf/focs/Lachish14}
Oded Lachish.
\newblock O(log log rank) competitive ratio for the matroid secretary problem.
\newblock In {\em 55th {IEEE} Annual Symposium on Foundations of Computer Science, {FOCS} 2014, Philadelphia, PA, USA, October 18-21, 2014}, pages 326--335, 2014.
\newblock \href {https://doi.org/10.1109/FOCS.2014.42} {\path{doi:10.1109/FOCS.2014.42}}.

\bibitem{Lindley1961DynamicPA}
Dennis~V. Lindley.
\newblock Dynamic programming and decision theory.
\newblock {\em Journal of The Royal Statistical Society Series C-applied Statistics}, 10:39--51, 1961.

\bibitem{DBLP:journals/mst/Ma0W16}
Tengyu Ma, Bo~Tang, and Yajun Wang.
\newblock The simulated greedy algorithm for several submodular matroid secretary problems.
\newblock {\em Theory Comput. Syst.}, 58(4):681--706, 2016.
\newblock URL: \url{https://doi.org/10.1007/s00224-015-9642-4}, \href {https://doi.org/10.1007/S00224-015-9642-4} {\path{doi:10.1007/S00224-015-9642-4}}.

\bibitem{oxley2006matroid}
James~G Oxley.
\newblock {\em Matroid Theory}.
\newblock Oxford University Press, 2006.

\bibitem{Schrijver}
Alexander Schrijver.
\newblock {\em Combinatorial optimization: polyhedra and efficiency}, volume~24.
\newblock Springer, 2003.

\bibitem{DBLP:journals/mor/SotoTV21}
Jos{\'{e}}~A. Soto, Abner Turkieltaub, and Victor Verdugo.
\newblock Strong algorithms for the ordinal matroid secretary problem.
\newblock {\em Math. Oper. Res.}, 46(2):642--673, 2021.
\newblock URL: \url{https://doi.org/10.1287/moor.2020.1083}, \href {https://doi.org/10.1287/MOOR.2020.1083} {\path{doi:10.1287/MOOR.2020.1083}}.

\end{thebibliography}

\end{document}